\documentclass[submission,copyright,creativecommons]{eptcs}
\providecommand{\event}{LSFA 2026} % Name of the event you are submitting to

\usepackage{iftex}

\ifpdf
  \usepackage{underscore}         % Only needed if you use pdflatex.
  \usepackage[T1]{fontenc}        % Recommended with pdflatex
\else
  \usepackage{breakurl}           % Not needed if you use pdflatex only.
\fi
\usepackage{packages}

\usepackage{mycommands}

\title{Efficient Decision Procedures for RNmatrix Semantics}
\author{
Renato R. Leme
\institute{Centre for Logic, Epistemology and The History of Science, UNICAMP, Brazil}
\email{rntreisleme@gmail.com}
    \and
    Carlos Olarte
\institute{Universit\'{e} Sorbonne Paris Nord, LIPN, CNRS, UMR 7030, F-93430, Villetaneuse, France}
\email{olarte@lipn.univ-paris13.fr}
\and
Elaine Pimentel
\institute{Department of Computer Science, University College London, UK}
\email{e.pimentel@ucl.ac.uk}
}

\def\titlerunning{Efficient Decision Procedures for RNmatrix Semantics}
\def\authorrunning{Renato R. Leme, Carlos Olarte \&  Elaine Pimentel}
\def\copyrightholders{R.R. Leme, C. Olarte \&  E. Pimentel}
\begin{document}
\maketitle

\begin{abstract}
%!TEX root = main.tex
% !TEX spellcheck = en-US

\sloppy 

Logical matrices provide a semantic framework in which connectives are interpreted by deterministic truth-functions. While elegant, this approach is often too restrictive to capture the semantics of many non-classical logics. Non-deterministic matrices (Nmatrices) generalise ordinary matrices by allowing connectives to return a set of possible truth values rather than a unique one, thereby increasing expressive power.
Restricted non-deterministic matrices (RNmatrices) further refine this framework by imposing constraints on the rows of Nmatrices, filtering out ``unsound'' rows and retaining only ``valid'' ones. This yields a more expressive semantic framework that has been successfully used to provide sound and complete semantics for a wide range of logics, including paraconsistent, intuitionistic, and modal logics.
Despite these advances, no {\em efficient} decision procedures based on RNmatrix semantics have been proposed. In this paper, we develop automated theorem provers based on RNmatrices by encoding their semantics as Satisfiability Modulo Theories (SMT) problems. The resulting provers decide validity and construct countermodels, achieving state-of-the-art performance for paraconsistent logics and competitive results for intuitionistic and modal logics.

\end{abstract}

\section{Introduction}\label{sec:intro}
%!TEX root = main.tex
% !TEX spellcheck = en-US

The principle of truth-functionality, central to classical and many-valued logics, states that the truth-value of a complex formula is uniquely determined by the truth-values of its subformulas. This makes truth tables both easy to understand and implement, enabling SAT solvers to crack remarkably hard problems. 
       
But this tidy picture cannot be obtained when crossing the boundary of classicality, reaching \eg, constructive logics. To illustrate the difficulty, suppose that propositional intuitionistic logic ($\ILp$) were given a three-valued semantics with values $\vT$ (true), $\vU$ (undetermined), and $\vF$ (false). In order to respect intuitionistic principles, the negation of the undetermined value $\vU$ would have to be $\vF$: lacking a proof of a statement certainly does not guarantee a proof of its impossibility. However, this would validate, \eg, the weak law of excluded middle $\neg \alpha \vee \neg\neg \alpha$, %as suggested by the table above, 
which is not derivable in intuitionistic logic.

More generally, in 1932 Kurt G\"{o}del showed that for any finite number $n$ of truth values, there exists a formula that an $n$-valued matrix would designate as true, but which is not valid in $\ILp$~\cite{Goedel32}. Intuitively, capturing the behaviour of intuitionistic truth requires accounting for the potentially unbounded evolution of proofs, which cannot be faithfully represented using only finitely many truth values.
%The problem is that the 3-valued table is too ``static,'' while intuitionistic  truth is not a single choice from a menu, but rather a pathway where validity means the possibility of finding a proof. 

From this, it is immediate that the same limitation applies to the modal logic $\mSfour$, since $\ILp$ can be faithfully embedded into $\mSfour$. In 1940, James Dugundji~\cite{DBLP:journals/jsyml/Dugundji40} extended this result to the entire modal cube, demonstrating that truth-functionality cannot, in general, be preserved when moving to modal notions of ``qualified truth.''
Finally, as noted in~\cite{DBLP:journals/logcom/AvronL05}, real-world information is often incomplete, uncertain, vague, imprecise, or inconsistent, and thus such a neat truth-functional treatment is also not always adequate. 

A way to recover truth tables as a semantic framework is to relax the notion of truth-functionality by replacing functions with multi-functions, thus allowing non-deterministic evaluations of formulas. This leads to {\em non-deterministic matrices (Nmatrices)}~\cite{Avron2011}, a generalization of ordinary multi-valued matrices in which the truth-value of a complex formula is chosen non-deterministically from a non-empty set of possible values.

Nmatrices offer a simple yet powerful semantic framework. By allowing multiple possible outcomes for formulas, they capture phenomena that deterministic matrices cannot, providing satisfactory semantic accounts for logics that lack finite deterministic characterizations. At the same time, their structure remains remarkably straightforward: they can be seen as a minimal generalization of familiar truth tables, where functions are replaced by multi-functions. This combination of simplicity and expressive flexibility has made Nmatrices a valuable tool, specially for paraconsistent logics~\cite{Avron2011}: they not only enable finite characterizations of otherwise intractable logics but also support compositional reasoning~\cite{DBLP:journals/tocl/CiabattoniLSZ14} and automated analysis~\cite{DBLP:journals/rsl/AvronZ19}.

Despite their generality, Nmatrices do not suffice to semantically characterize $\ILp$, da Costa's paraconsistent logic $C_1$~\cite{DBLP:journals/ndjfl/Costa74a} and normal modal logics, as these logics cannot be captured by any single finite-valued non-deterministic matrix.
In the case of $\ILp$, while finite-valued tables for $\ILp$ are too ``static'', in the sense that they cannot capture the evolving proof-oriented nature of intuitionistic validity, introducing non-determinism swings too far in the opposite direction, becoming overly permissive and allowing rows that do not faithfully reflect the intended behavior of logical connectives\footnote{See~\cite[Theorem 3.5]{leme2025newIPL} for $\ILp$; \cite[Theorem 11]{DBLP:journals/ijar/Avron07} for paraconsistent logics and~\cite{DBLP:journals/logcom/Gratz22} for the modal cube.}. 
For example, if the negation of the undetermined value $\vU$ can assume the values in $\{\vU,\vT\}$, then there is a valuation assigning to the intuitionistic tautology $\neg\neg (\alpha\vee\neg\alpha)$ the non-designated value $\vU$.

{\em Restricted non-deterministic matrices (RNmatrices)} %address this limitation 
then swings back to find the perfect balance on non-determinism by imposing constraints on admissible rows, thereby filtering out unsound interpretations while preserving the expressive power of Nmatrices.  It was formally introduced in~\cite{coniglio2022two} as a generalization of restricted matrices (Rmatrices) of Brunetto Piochi~\cite{piochi1978matrici,piochi1983logical} to give a semantic characterization with decision procedures for the entire hierarchy of da Costa's calculi $C_n$. %\rl{reorganizei aqui, trouxe pra cima}

Nonetheless, restricting the set of valuations has been considered before. For instance, John Kearns introduced constraints on valuations known as {\em level valuations}~\cite{DBLP:journals/jsyml/Kearns81}  to characterize the modal systems $\mKT$, $\mSfour$, and $\mSfive$. The idea is that level $0$ corresponds to the standard evaluation induced by the underlying matrix, while higher levels progressively filter out valuations that violate the modal rule of necessitation, retaining only those that assign designated values to tautologies. Validity is then defined as the (generally infinite) intersection of all such levels, capturing closure under necessitation. %This method has been extended to paraconsistent logics in~\cite{coniglio2022two}, where the entire hierarchy of da Costa's calculi $C_n$ is given a semantic characterization, yielding a novel decision procedure for these logics.

Although conceptually interesting, this idea has limited practical applicability, as determining level-valuations requires accounting for the valuations of {\em all} tautologies across {\em all} levels. This requires a twofold infinite testing: on formulas and levels.
In order to tackle this problem, Lukas Gr\"{a}tz~\cite{DBLP:journals/logcom/Gratz22} introduced a \emph{decision procedure} for Nmatrices that is both sound and complete with respect to Kearns' semantics. 

This development enabled the proposal of remarkably simple and effective decision procedures to broad class of logics.
%extension of the methods of Kearns and Gr\"{a}tz to a broader class of logics.
%Building on these ideas, recent work has shown that RNmatrices can yield remarkably simple and effective decision procedures. 
For instance,~\cite{leme2025newIPL} presents a semantic characterization of $\ILp$ using three-valued RNmatrices, leading to a lightweight decision procedure for intuitionistic propositional logic; in~\cite{DBLP:conf/cade/LahavZ22}, the semantic analyticity of Kearns' semantics is proved, providing an alternative decidability proof to Gr\"{a}tz's result for $\mK$ and $\mKT$; and in~\cite{DBLP:conf/tableaux/LemeOPC25} we developed RNmatrix-based decision procedures for all 15 logics in the modal cube.

In this way, RNmatrices retain the conceptual simplicity and flexibility of Nmatrices while providing a more precise and robust semantic framework.
%
%Crucially, RNmatrices preserve the essential simplicity of truth tables. 
This makes them particularly well-suited as a foundation for automated reasoning: they combine the intuitive, tabular nature of truth-table semantics with the expressive power required to handle non-classical phenomena. % , enabling effective and scalable computational methods. 

However, the application of these semantics in the context of automatic theorem provers has been elusive. In particular, a naive implementation of the matrix semantics is limited to very small formulas due to the inherent combinatorial explosion of the state space. Hence, the next natural step, which we pursue in this paper, is to leverage this structure within modern Satisfiability Modulo Theories (SMT) solvers, thereby bridging the gap between elegant semantic frameworks and high-performance automated reasoning tools.

Inspired by science fiction, we present $\tool$~\cite{tool} (Theorem prover for
RNmatrices), a general framework for building automated theorem provers for
logics semantically described ``in the \texttt{matrix}''. By encoding RNmatrices and their
elimination criteria as SMT problems, we leverage off-the-shelf SMT solvers to
decide formula validity and construct countermodels. We demonstrate the
approach on paraconsistent logics, where our prover outperforms the current
state of the art and provides the first implementation covering the entire
$C_n$ hierarchy\footnote{We note that~\cite{neto2006effective} presents a
strategy for an automatic theorem prover based on a tableaux method for $C_1$.
Moreover, the RNmatrix semantics introduced in~\cite{coniglio2022two} naturally
induces a labelled tableau system for each $C_n$, yielding an alternative
decision procedure for these logics. Finally,~\cite{d2006analytical} also
provides decision procedures for da Costa's hierarchy based on tableaux.}. We
further evaluate the framework on intuitionistic and modal logics, where it
achieves competitive performance.
Finally, we benchmark $\tool$ against existing systems for these logics, obtaining consistently strong results.

\begin{comment}
Highlighting the contributions:
\begin{itemize}
    \item RNmatrices are more expressive to Nmatrices and they provide a ``natural'' way of
        given semantics to non-classical logics. \magenta{EP. Done!}
    \item Recently several RNmatrix semantics have been proposed, for intuitionistic, 
        modal and paraconsistent logcs. \magenta{EP. Done!}
    \item However, the application of these semantics in the context of automatic theorem provers
        has been elusive. In particular, a naive implementation of the matrix semantics
        is doomed to prove very small formulas to the inherent combinatorial state stace. \magenta{EP.Done!}
    \item This paper shows how RNmatrices can be efficiently encoded as an SMT problem, 
        thus opening the possibility of defining new automatic theorem provers for 
        all the above logics. In particular, 
        we obtain the first\footnote{As far as we know, although they can be also implemented in 
            other generic frameworks as X and Y that generate provers out of sequent or tableaux systems.}  theorem prover for $C_n$ ($n\geq 2$) is proposed. 
    \item We benchmark our prover with state-of-the-art ones for those logics,
        obtaining very competitive results. 
\end{itemize}
\end{comment}

\section{RNmatrix Semantics}\label{sec:rnmat}
%!TEX root = main.tex
% !TEX spellcheck = en-US

Starting from the well known 2-valued matrix semantics for propositional classical logics, consider the semantic tables for negation $\neg$ and  disjunction $\vee$:
\begin{center}
        \begin{tabular}{c @{\hspace{2em}} c}
        \toprule
        \textbf{$\alpha$} &  \textbf{$\neg \alpha$} \\
        \midrule
        $\vT$ & $\vF$ \\
        $\vF$ & $\vT$ \\
        \bottomrule
    \end{tabular}
  \qquad \qquad
         \begin{tabular}{c @{\hspace{2em}} c @{\hspace{2em}} c}
        \toprule
        \textbf{$\alpha$} &  \textbf{$\beta$} &  \textbf{$\alpha\vee\beta$}\\
        \midrule
        $\vT$ & $\vT$  & $\vT$\\
        $\vT$ & $\vF$  & $\vT$\\
        $\vF$ & $\vT$ &  $\vT$\\
        $\vF$ & $\vF$ &  $\vF$\\
        \bottomrule
    \end{tabular}
%      \begin{tabular}{c c @{\hspace{2em}} c}
%        \toprule
%        \textbf{$P$} & \textbf{$\neg P$} & \textbf{$\neg P \vee \neg\neg P$} \\
%        \midrule
%        T & T & T \\
%        U & F & T \\
%        F & T & T \\
%        \bottomrule
%        \end{tabular}
\end{center} 

\noindent
As noted in~\cite{Avron2011}, if one wishes to reject the law of excluded middle $\alpha\vee\neg\alpha$, it is sufficient to discard  the information concerning the second line of the truth-table for $\neg$. 
However, this leads to a problem of \emph{underspecification}, where the meaning of a connective is not fully determined.

One way to address this, while still rejecting excluded middle, is to interpret the value of $\neg\alpha$ at $\vF$ {\em non-deterministically}, allowing it to take any value in $\{\vT,\vF\}$ (table below left). This results in the following truth-table for $\alpha\vee\neg\alpha$ (below right):
\begin{center}
        \begin{tabular}{c @{\hspace{2em}} c}
        \toprule
        \textbf{$\alpha$} &  \textbf{$\neg \alpha$} \\
        \midrule
        $\vT$ & $\vF$ \\
        $\vF$ & \{$\vT$,$\vF$\} \\
        \bottomrule
    \end{tabular}
     \qquad \qquad
         \begin{tabular}{c @{\hspace{2em}} c @{\hspace{2em}} c}
        \toprule
        \textbf{$\alpha$} &  \textbf{$\neg\alpha$} &  \textbf{$\alpha\vee\neg\alpha$}\\
        \midrule
        $\vT$ & $\vF$  & $\vT$\\
        $\vF$ & $\vT$ &  $\vT$\\
        $\vF$ & $\vF$ &  $\vF$\\
        \bottomrule
    \end{tabular}
\end{center}
The last row shows that $\alpha\vee\neg\alpha$ needs not evaluate to $\vT$, and hence it is not a tautology. This provides the basic intuition behind non-deterministic matrices, where single semantic values are replaced by {\em sets} of possible values or, equivalently, functions are replaced by {\em multifunctions}.
%\footnote{Observe that, in the case of $\ILp$, this leads to an {\em overspecification}, }.

Formally, a non-deterministic matrix (Nmatrix), introduced by Avron and Lev in~\cite{DBLP:journals/logcom/AvronL05}, generalizes a logical matrix by incorporating the notion of non-deterministic truth-functions, or multifunctions, in the semantics. Such functions are mappings from tuples of truth-values to non-empty sets of truth-values, which represents the acceptable interpretations of a complex formula in a given a context. 

In what follows, given a propositional language $\Lan$, $\ForL{\Lan}$ denotes the set of well-formed formulas in $\Lan$, $\alpha,\beta\ldots$ (resp. $\Delta,\Gamma,\Lambda\ldots$) range over 
elements of $\ForL{\Lan}$ (resp. $\wp(\ForL{\Lan})$) and $p,q\ldots$ represent atomic propositions.
%and for any set $\Gamma$, $\wp(\Gamma)$ is the powerset of $\Gamma$. %, \ie, the set of all sets $\Gamma' \subseteq \Gamma$. 
Given a formula $\alpha$, we denote the set containing every subformula of $\alpha$ by $sub(\alpha)$. 

\begin{definition}[Nmatrix]\label{def:nmatrix}
%An {\em  Nmatrix} for $\Lan$ is a tuple $\M = \langle \V,\D,\Om\rangle$, where:
An {\em  Nmatrix} for $\Lan$ is a tuple $\ML{\Lan} = \langle \V,\D,\Om\rangle$, where:
\begin{itemize}
\item $\V$ is a non-empty set of truth-values.
\item $\D$ (designated truth-values) is a non-empty proper subset of $\V$.
\item For every n-ary connective $\conn$ in $\Lan$, $\Om$ includes a  non-deterministic truth-function $\tilde\conn: \V^n \to \wp(\V)\backslash\varnothing$.
%\V^n \to 2^{\V}\backslash\varnothing$.
\end{itemize}
%We say that P is (in)finite if so is V.
\end{definition}

\begin{definition}[Valuation]\label{def:val}
Let $\ML{\Lan} = \langle \V,\D,\Om\rangle$ %$\M = \langle \V,\D,\Om\rangle$ 
be an Nmatrix and $\Lambda\subseteq \ForL{\Lan}$  closed under subformulas.
A {\em partial valuation} in $\ML{\Lan}$ is a function $\vv:\Lambda\to \V$ such that, for each n-ary connective $\conn$ in  $\Lan$, the following holds for all $\alpha_1, \ldots, \alpha_n\in\Lambda$:
$
   v (\conn (\alpha_1,\ldots, \alpha_n)) \in \tilde\conn (v(\alpha_1), \ldots, v(\alpha_n))
   $.
A partial valuation in $\ML{\Lan}$ is a (total) {\em valuation} if its domain is $\ForL{\Lan}$. 

Given an Nmatrix $\ML{\Lan}$,  $Val(\ML{\Lan})$ is  the set of {\em all valuations} of $\ML{\Lan}$.
\end{definition}

%\begin{remark}
%Given an Nmatrix $\mathbb{M}$, we denote the set of all valuations of $\mathbb{M}$ by $Val(\mathbb{M})$.
%\end{remark}

It turns out that a valuation function makes no distinction between the different outputs of a non-deterministic truth-function. Some logics accommodate this, such as classical logic or some modal logics without the rule of necessitation, while others do not. In particular, intuitionistic logic, normal modal logics and da Costa's paraconsistent logic $C_1$ cannot be characterized by finite Nmatrices (see~\cite{leme2025newIPL,DBLP:journals/logcom/Gratz22} and~\cite{DBLP:journals/ijar/Avron07} respectively). 
%it is known that all the logics covered by our tool $\tool$ 
This issue can be addressed, while preserving finiteness of the set of truth-values, by refining the set of valuations.

RNmatrix, formally introduced in~\cite{coniglio2022two}, 
controls the outputs of Nmatrix in cases where blind handling of non-determinism is not possible. In such cases, contextual information imposes a revision of the valuation assignment, otherwise unsoundness arises. As has been shown by recent works (see \eg~\cite{EConiglio2025-CONRFM}), several logics that are not characterizable by finite Nmatrices can instead be captured by imposing suitable constraints on the outputs of non-deterministic truth-functions.

\begin{definition}[RNmatrix]\label{def:rnmat}
    An RNmatrix is a pair $\mathcal{R}(\ML{\Lan}) = \langle \mathcal{F},
\ML{\Lan}\rangle$ where $\ML{\Lan}$ is an Nmatrix and $\mathcal{F} \subseteq
Val(\ML{\Lan})$ is a subset of the set of valuations in $\ML{\Lan}$. 
\end{definition}

For every RNmatrix considered in this paper, the set $\mathcal{F}$ can be specified as a sentence in first-order logic. This will be useful for the encoding presented in~\Cref{sec:encoding}.  Currently, our tool $\tool$ (Theorem prover for RNmatrices)~\cite{tool}  specifies the RNmatrices of modal logic $\mSfour$, intuitionistic propositional logic ($\ILp$), and the family of paraconsistent logics $\cn$. In what follows, we briefly present each of them.

%Intuitively, the idea is that the valuations defined by the Nmatrix form the basis for the logic at hand, while the restriction identifies the precise subset of valuations that characterize it. In the following section, the restriction will be formalized in a first order logic. %Such a distinction motivates the introduction of a restricted consequence relation.

% \begin{definition}[Semantic consequence]\label{def:conseq} Let $\Gamma\cup \{ A \}$ be a set of formulas and $\mathbb{M} = \langle \mathcal{V}, \mathcal{D}, \mathcal{O} \rangle$ an Nmatrix. $\Gamma \models^{\mathbb{M}} A$ if and only if, for every $v \in Val(\mathbb{M})$, if $v(\Gamma) \in D$, then $v(A) \in D$. We use $\models^{\mathbb{M}}_{\Delta}$ to denote the consequence relation restricted to a subset $\Delta \subseteq Val(\mathbb{M})$.
% \end{definition}

\begin{table}
\centering
\resizebox{\textwidth}{!}{
\begin{tabular}{llll}
\hline
Logic & Language & Truth-values & Designated values \\
\hline

$\mSfour$ &
$\alpha ::= p \mid \bot \mid \top \mid \neg \alpha \mid \Box\alpha  \mid \alpha \to \alpha \mid \alpha\land\alpha \mid \alpha\lor\alpha$ &
$\V_{\mSfour}=\{0, 1, 2\}$ &
$\D_{\mSfour}=\{1,2\}$ \\[6pt]

%$\mathbb{M}_{\mSfive} $ &
%$\alpha ::= p \mid \neg \alpha \mid \Box\alpha \mid \Diamond \alpha \mid \alpha \to \alpha \mid \alpha\land\alpha \mid \alpha\lor\alpha$ &
%$\V_{\mSfive}=\{\vT,\vt,\vf,\vF\}$ &
%$\D_{\mSfive}=\{\vT,\vt\}$ \\[6pt]

$\ipl$ &
$\alpha ::= p \mid \bot \mid \top \mid \neg \alpha \mid \alpha \to \alpha \mid \alpha\land\alpha \mid \alpha\lor\alpha$ &
$\V_{\ipl}=\{\vT,\vF\}$ &
$\D_{\ipl}=\{\vT\}$ \\[6pt]

$\C{n}$ &
$\alpha ::= p \mid \neg\alpha \mid \alpha \to \alpha \mid \alpha\land\alpha \mid \alpha\lor\alpha$ &
$\V_{\cn}=\{\vTn,\vtn{0},\ldots,\vtn{n-1},\vFn\}$ &
$\vDn=\V_{\cn}\setminus\{\vFn\}$ \\

\hline
\end{tabular}
}
\caption{Nmatrices considered in this paper. Multifunctions are given in
Tables~\ref{table:CnNmatrices}, \ref{table:IPLNmatrices} and \ref{table:modalNmatrices}.}
\label{table:allNmatrices}
\end{table}

\subsection{da Costa's family $\cn$}\label{mat:cn}

In classical logic, inconsistent theories are trivial, since the {\em explosion principle} $\alpha \to (\neg\alpha \to \beta)$ allows arbitrary conclusions to be derived from a contradiction.
By contrast, paraconsistent approaches reject explosion, allowing inconsistencies to be handled in a non-trivial way.
This is particularly relevant in contexts such as databases, which may contain conflicting information about the same entity while still preserving meaningful content, or in modeling the set of {\em beliefs} of an agent: it is reasonable to allow both truth and falsity to coexist without collapse, since agents may hold inconsistent beliefs without thereby accepting everything.

The family of logics $C_n$ (for $1 \le n < \omega$) constitutes a hierarchy of paraconsistent systems introduced by Newton da Costa~\cite{daCosta1993Sistemas}.  These systems are designed to progressively constrain the principle of explosion  by increasingly weakening the role of the principle of non-contradiction. The systems are ordered by strict inclusion, moving from the most rigid (classical logic) to increasingly permissive paraconsistent systems:
\[
C_\omega \subset \ldots \subset C_{n+1} \subset C_{n} \subset  \ldots \subset C_2 \subset C_1 \subset C_0 \cong  {\CLp}
\]
where $C_0$ corresponds to classical propositional logic $\CLp$. In this hierarchy, each logic $C_{i+1}$ is strictly weaker than $C_i$ because it tolerates a higher ``depth'' of contradiction before collapsing into triviality.

To formalize this depth, the degree of consistency of a formula $\alpha$ is defined recursively as
$\alpha^0 := \alpha$ and, for $n \geq 0$:
\[
\alpha^{n+1} := \neg(\alpha^n \land \neg \alpha^n)
\]
We define the $k$-depth of contradiction as $\mathtt{c}_k(\alpha) \equiv \alpha^k \land \neg \alpha^k$. Essentially, $\mathtt{c}_0$ represents a simple contradiction ($\alpha \land \neg \alpha$), while $\mathtt{c}_1$ represents a contradiction regarding the law of non-contradiction itself ($\alpha^1 \land \neg \alpha^1$).

As noted in~\cite{CarnielliConiglio2016mbC}, the system $C_n$ is calibrated to prevent explosion for any contradiction depth $k$ less than $n$. However, it remains ``explosive'' once that specific threshold is reached:

\begin{remark}\label{remark:cn}
    For every $n > 0$ and any formulas $\alpha, \beta$:
    \[
    \mathtt{c}_{k}(\alpha) \not\vdash_{C_n} \beta \quad \text{for } 0 \le k < n
    \]
    \[
    \mathtt{c}_{k}(\alpha) \vdash_{C_n} \beta \quad \text{for } k \geq n
    \]
\end{remark}
This illustrates that $C_n$ behaves paraconsistently at lower levels of contradiction, but retains a ``classical core'' that triggers explosion when the contradiction reaches the $n$-th level of nested consistency.

The $\cn$-bivaluations proposed by da Costa and Alves in~\cite{da1977semantical} and fully developed by Lopari\'{c} and Alves in~\cite{loparic1980semantics} are employed by Coniglio and Toledo in~\cite{coniglio2022two} to obtain a decision method for the $\cn$ hierarchy based on restricted finite non-deterministic matrices. To this end, each truth-value of a formula $\alpha$ in $\cn$ is defined as a tuple (called a snapshot), where each position encodes the bivaluation of $\alpha$, $\neg \alpha$, and of $\alpha^1,\ldots,\alpha^n$. Formally, let $\mathtt{b}$ be a $\cn$-bivaluation. Then, a truth-value is an $(n+2)$-tuple, where
\[
\langle \mathtt{b}(\alpha), \mathtt{b}(\neg \alpha), \mathtt{b}(\alpha^1), \ldots,  \mathtt{b}(\alpha^n) \rangle
\]

\begin{table}[t]
\centering
\small
\setlength{\tabcolsep}{4pt}

\vspace{6pt}
\resizebox{\textwidth}{!}{
\begin{minipage}{0.18\textwidth}
\centering
\begin{tabular}{c|c}
$x$ & $\neg^{\cn} x$ \\ \hline
$\vTn$ & $\vFn$ \\ 
$\vti$ & $\vDn$ \\
$\vFn$ & $\vTn$ \\
\end{tabular}
\end{minipage}
\hspace{0.8em}
\begin{minipage}{0.24\textwidth}
\centering
\begin{tabular}{c|ccc}
$x\to^{\cn}y$ & $\vTn$ & $\vtj$ & $\vFn$ \\ \hline
$\vTn$ & $\vTn$ & $\vDn$ & $\vFn$ \\ 
$\vti$ & $\vDn$ & $\vDn$ & $\vFn$ \\ 
$\vFn$ & $\vTn$ & $\vDn$ & $\vTn$ \\
\end{tabular}
\end{minipage}
\hspace{0.8em}
\begin{minipage}{0.24\textwidth}
\centering
\begin{tabular}{c|ccc}
$x\land^{\cn}y$ & $\vTn$ & $\vtj$ & $\vFn$ \\ \hline
$\vTn$ & $\vTn$ & $\vDn$ & $\vFn$ \\
$\vti$ & $\vDn$ & $\vDn$ & $\vFn$ \\
$\vFn$ & $\vFn$ & $\vFn$ & $\vFn$ \\ 
\end{tabular}
\end{minipage}
\hspace{0.8em}
\begin{minipage}{0.24\textwidth}
\centering
\begin{tabular}{c|ccc}
$x\lor^{\cn}y$ & $\vTn$ & $\vtj$ & $\vFn$ \\ \hline
$\vTn$ & $\vTn$ & $\vDn$ & $\vTn$ \\ 
$\vti$ & $\vDn$ & $\vDn$ & $\vDn$ \\ 
$\vFn$ & $\vTn$ & $\vDn$ & $\vFn$ \\
\end{tabular}
\end{minipage}
}
\caption{Non-deterministic truth-functions for $\mathbb{M}_{\cn}$. }\label{table:CnNmatrices}
\end{table}

In general, each $\mathbb{M}_{\cn}$ contains $n$ inconsistent truth-values together with two consistent (classical) ones. For instance, $\mathbb{M}_{C_1}$ has three truth-values: $T_1 = \langle 1,0,1 \rangle$, $F_1 = \langle 0,1,1 \rangle$, and $t^1_0 = \langle 1,1,0 \rangle$. \Cref{table:allNmatrices} presents the language, the set of truth-values and the respective designated values of each $\cn$. The non-deterministic truth-functions for each $\cn$ is presented in~\Cref{table:CnNmatrices}.

Now suppose that $v(\alpha) = t^{1}_0$ in $\mathbb{M}_{C_1}$. This valuation represents a situation in which both $\alpha$ and $\neg \alpha$ are true at the level of the underlying bivaluation. Still, the truth-value assigned to $\alpha \land \neg \alpha$ is non-deterministic, branching between $T_1$ and $t^1_0$. As a consequence, some resulting valuations validate the principle of non-contradiction (see~\Cref{table:exc1}). Such valuations are incompatible with the intended bivaluation semantics captured by the snapshot and should therefore be excluded. To solve this problem, the restriction selects only the valuations $v'$ such that $v'(\alpha \land \neg \alpha) = T_1$. This reasoning is then generalized to all the logics in the hierarchy as follows, where $I_n := \{ t^n_i \mid 0 \leq i < n \}$.

\begin{table}
\centering
\begin{tabular}{cccc}
\toprule
$\alpha$ & $\neg \alpha$ & $\alpha \land \neg\alpha$ & $\alpha^1$ \\
\midrule
$t^1_0$ & $T_1$   & $T_1$   & $F_1$ \\
$t^1_0$ & $T_1$   & $t^1_0$ & $\{T_1, t^1_0\}$ \\
$t^1_0$ & $t^1_0$ & $T_1$   & $F_1$ \\
$t^1_0$ & $t^1_0$ & $t^1_0$ & $\{T_1, t^1_0\}$ \\
\bottomrule
\end{tabular}
\caption{Example of Table in $\mathbb{M}_{C_1}$.}\label{table:exc1}
\end{table}

\begin{definition}[$\mathcal{F}_{\cn}$]\label{def:rcn}
    $\mathcal{F}_{\cn}$ is the set of every  $v \in Val(\mathbb{M}_{\cn})$ such that:
    \begin{enumerate}
        \item If $v(\alpha) = \vtn{0}$, then $v(\alpha \land \neg \alpha) = T_n$;
        \item For every $1 \leq k < n$, if $v(\alpha) = \vtn{k}$, then $v(\alpha \land \neg \alpha) \in I_n$ and $v(\alpha^1) = \vtn{k-1}$.
    \end{enumerate}
\end{definition}

\begin{definition}[RNmatrix for $\cn$]
    Given $n>0$, $\mathcal{R}(\mathbb{M}_{\cn}) = \langle \mathcal{F}_{\cn}, \mathbb{M}_{\cn} \rangle$ is the RNmatrix for $\cn$.
\end{definition}

\subsection{$\ipl$}\label{sec:ipl}

G\"{o}del showed in~\cite{Goedel32} that rejecting the principle of excluded middle entails that no single finite logical matrix characterizes $\ipl$. The same holds for finite non-deterministic matrices, as shown in~\cite{leme2025newIPL}. However, as with the paraconsistent logics discussed above, this limitation can be overcome by imposing suitable restrictions on an appropriate Nmatrix. To this end, Leme, Coniglio, and Lopes~\cite{leme2025newIPL} introduced a $3$-valued RNmatrix for $\ipl$ based on the truth-values \emph{proved} ($\vT$), \emph{not proved} ($\vF$) and \emph{neither proved nor refuted} ($\vU$). It turns out that this semantics can be reduced to a $2$-valued equivalent semantics.
$\tool$ implements the $2$-valued version, which we now briefly describe.

In this semantics, \emph{proved} ($\vT$) and \emph{not proved} ($\vF$) are taken as the only primitive values. The status of the components determines the value of a complex proposition. For instance, if $\alpha$ is proved and $\beta$ is not proved, then under a constructive reading of conjunction and disjunction, $\alpha \lor \beta$ is proved, whereas $\alpha \land \beta$ is not proved. Moreover, since $\alpha$ is proved while $\beta$ is not, $\alpha \to \beta$ is not proved. This yields the non-deterministic truth-functions shown in~\Cref{table:IPLNmatrices}.

Note that, when both $\alpha$ and $\beta$ are not proved, the value of the complex formula $\alpha \to \beta$ is not fully determined. In this situation, we cannot decide whether $\alpha \to \beta$ is proved. For instance, let $\alpha$ be a conjecture and $\beta = \bot$. By assumption, we have no proof of $\alpha$ nor of $\beta$. Moreover, we have no proof of $\alpha \to \bot$, since $\alpha$ is not refuted either. Thus, the truth-value of $\alpha \to \beta$ remains open. Accordingly, the truth-function allows both $\vF$ and $\vT$ as possible values for $\alpha \to \beta$. The resulting Nmatrix, however, is weaker than $\ipl$. To recover the intended strength, we impose a restriction on the set of valuations.

\begin{table}[t]
\centering
\small
\setlength{\tabcolsep}{4pt}
\resizebox{\textwidth}{!}{
\begin{tabular}{c@{\qquad}c@{\qquad}c@{\qquad}c@{\qquad}c}
\multicolumn{5}{c}{} \\[4pt]

\begin{tabular}{c|c}
$\bot^{\ipl}$ & $\top^{\ipl}$\\ \hline
$\vF$ & $\vT$ \\ 
\end{tabular}
&
\begin{tabular}{c|c}
$x$ & $\neg^{\ipl} x$ \\\hline
$\vF$ & $\vF,\vT$  \\ 
$\vT$ & $\vF$ \\ 
\end{tabular}
&
\begin{tabular}{c|cc}
$x\to^{\ipl}y$ & $\vF$ & $\vT$ \\\hline
$\vF$ & $\vF,\vT$ & $\vT$ \\ 
$\vT$ & $\vF$ & $\vT$ \\ 
\end{tabular}
&
\begin{tabular}{c|cc}
$x\lor^{\ipl}y$ & $\vF$ & $\vT$ \\ \hline
$\vF$ & $\vF$ & $\vT$ \\
$\vT$ & $\vT$ & $\vT$ \\ 
\end{tabular}
&
\begin{tabular}{c|cc}
$x\land^{\ipl}y$ & $\vF$ & $\vT$ \\ \hline
$\vF$ & $\vF$ & $\vF$ \\ 
$\vT$ & $\vF$ & $\vT$ \\ 
\end{tabular}
\\[10pt]

% \multicolumn{4}{c}{} \\[4pt]

% \begin{tabular}{|c|}
% \hline
% $\bot^{\ipl}_3$ \\ \hline
% $\vF$ \\ \hline
% \end{tabular}
% &
% \begin{tabular}{|c|c|c|c|}
% \hline
% $z\to^{\ipl}_3y$ & $\vF$ & $\vU$ & $\vT$ \\ \hline
% $\vF$ & $\vU,\vT$ & $\vU,\vT$ & $\vT$ \\ \hline
% $\vU$ & $\vU,\vT$ & $\vU,\vT$ & $\vT$ \\ \hline
% $\vT$ & $\vF$ & $\vF$ & $\vT$ \\ \hline
% \end{tabular}
% &
% \begin{tabular}{|c|c|c|c|}
% \hline
% $z\lor^{\ipl}_3y$ & $\vF$ & $\vU$ & $\vT$ \\ \hline
% $\vF$ & $\vF$ & $\vF$ & $\vT$ \\ \hline
% $\vU$ & $\vF$ & $\vF$ & $\vT$ \\ \hline
% $\vT$ & $\vT$ & $\vT$ & $\vT$ \\ \hline
% \end{tabular}
% &
% \begin{tabular}{|c|c|c|c|}
% \hline
% $z\land^{\ipl}_3y$ & $\vF$ & $\vU$ & $\vT$ \\ \hline
% $\vF$ & $\vF$ & $\vF$ & $\vF$ \\ \hline
% $\vU$ & $\vF$ & $\vF$ & $\vF$ \\ \hline
% $\vT$ & $\vF$ & $\vF$ & $\vT$ \\ \hline
% \end{tabular}

\end{tabular}
}

\caption{Non-deterministic truth-functions for $\mathbb{M}_{\ipl}$.}
\label{table:IPLNmatrices}
\end{table}

\paragraph{Restricting the values} Constructively, an implication is false whenever there exists a proof of $\alpha$ that cannot be transformed into a proof of $\beta$. This corresponds to a valuation $w$ such that $w(\alpha) = \vT$ and $w(\beta) = \vF$. Recall that the non-deterministic case arises when neither $\alpha$ nor $\beta$ is proved. In such a situation, one cannot immediately conclude that $\alpha \to \beta$ is false. However, if there exists another valuation in which $\alpha$ is proved while $\beta$ is not, then this valuation can be identified by inspecting the matrix and used to justify that $\alpha \to \beta$ is false in the original valuation via monotonicity. 

\begin{remark}[Monotonicity]
    If $\Gamma, \Delta \not\vdash \gamma$, then $\Gamma \not\vdash \gamma$.
\end{remark}

To this end, the only requirement is that the new valuation may introduce additional proofs ($\Delta$), but must agree with the original valuation on all propositions already proved ($\Gamma$). Intuitively, this reflects the idea that mathematical knowledge is persistent and open-ended. Given that $\neg \alpha$ is defined as $\alpha \to \bot$, a similar reasoning is applied to negation.

\begin{definition}[$\mathcal{L}^{\ipl}$]\label{def:ripl}
    Let $\mathcal{L}^{\ipl}_0 = Val(\mathbb{M}_{\ipl})$. $\mathcal{L}^{\ipl}_{k+1}$ is the set of every $v \in \mathcal{L}^{\ipl}_{k}$ such that, for every $\alpha,\beta$,
    \begin{description}
        \item[$\to$] If $v(\alpha)=v(\beta)=v(\alpha\to \beta) = \vF$, then there is $w \in \mathcal{L}^{\ipl}_{k}$ such that $w(\alpha) = \vT$ and $w(\beta) = \vF$ and, for every $\gamma$, $v(\gamma) = \vT$ implies $w(\gamma) = \vT$.
        \item[$\neg$]  If $v(\alpha)=v(\neg \alpha)= \vF$, then there is $w \in \mathcal{L}^{\ipl}_{k}$ such that $w(\alpha) = \vT$ and, for every $\gamma$, $v(\gamma) = \vT$ implies $w(\gamma) = \vT$.
    \end{description}
    The restricted set of valuations is then defined as follows:
    \[
    \mathcal{L}^{\ipl} = \bigcap^{\infty}_{i=0} \mathcal{L}^{\ipl}_i
    \]
\end{definition} 

In contrast to the restriction defined for $\mathbb{M}_{\cn}$ in the previous section, the restriction for $\mathbb{M}_{\ipl}$ requires the existence of an additional valuation. Nevertheless, this existence can be verified locally.

Note that the criterion is applied backward rather than forward.  Given any formula $\gamma$, the values that can be assigned to $v(\gamma)$ (i.e., the entries in the column of $\gamma$ in the Nmatrix) depend only on the values of its subformulas. On the one hand, it is not difficult to see that if the existence criterion fails for all partial valuations whose domain is restricted to the set of subformulas of $\gamma$, then it also fails for all extensions of such partial valuations. In other words, if no suitable partial valuation exists locally, then no valuation exists globally. On the other hand, if the existence of a valuation can be verified locally, then, by the analyticity result in~\cite{leme2025newIPL}, this valuation can always be extended to a full valuation. Therefore,~\cref{def:ripl} supports decision procedures.

\begin{remark}
It is worth noting that the valuation functions in the restricted set can be interpreted as worlds in Kripke semantics, both in $\ipl$ as in $\mSfour$. From this perspective, the existence criteria induce an accessibility relation between different possible worlds. This interpretation is secondary but useful, as it imports known results for each logic, for example, regarding bounds on the models.
\end{remark}

% \begin{remark}
%     Given a formula $\alpha$, we denote the set of every subformula of $\alpha$ by $sub(\alpha)$. 
% \end{remark}

\begin{definition}[RNmatrix for $\ipl$]
    The RNmatrix for $\ipl$ is given by $\mathcal{R}(\mathbb{M}_{\ipl}) = \langle \mathcal{L}^{\ipl}, \mathbb{M}_{\ipl} \rangle$.
\end{definition}

\subsection{$\mSfour$}\label{sec:modal}

$\mSfour$ is a modal logic characterized by the following principles. First, $\Box$ is distributive over implication: $\Box (\alpha \to \beta) \to (\Box \alpha \to \Box \beta)$ (axiom $\mk$). Second, tautologies are necessarily true: from $\alpha$ infer $\Box \alpha$ (rule of necessitation). Third, necessity implies truth: $\Box \alpha \to \alpha$ (axiom $\mt$). Finally, $\mSfour$ collapses iterated necessity, as necessity is necessarily necessary: $\Box \alpha \to \Box \Box \alpha$ (axiom $\mfour$).

Modalities, such as $\Box$, add a new layer of complexity to the analysis of the concept of truth. Under the alethic interpretation, a proposition may be not only true but necessarily true, as in the case of mathematical truths; it may also be contingently true, as in everyday situations. Falsity can also be qualified: a proposition may be contingently false, but it can also be impossible. From a many-valued perspective, this suggests that \emph{truth} can be split into several different notions, such as necessarily true, impossibly false, or contingently true or false. In this perspective, different modal logics may accept or reject different truth-values. For example, $\mSfour$ would accept a truth-value that means true and necessarily true, but not a value meaning false and necessarily true.

Nonetheless, since $\ipl$ cannot be characterized by finite Nmatrices but can be translated into $\mSfour$, the same limitation applies to $\mSfour$. In the modal case, the problem stems from the rule of necessitation: for every Nmatrix that is complete with respect to $\mSfour$, say $\mathbb{M}_{\mSfour}^*$, there exists a formula $\alpha$ and a valuation $v \in Val(\mathbb{M}_{\mSfour}^*)$ such that $\alpha$ is a tautology, yet $v(\Box \alpha)$ is non-designated. In other words, the difficulty once again lies in establishing the soundness of the system. Kearns addressed this problem by defining validity in terms of a restricted set of valuations using a technique now known as \emph{level valuations} (see~\cite{DBLP:journals/jsyml/Kearns81}). However, level valuations do not directly yield a decision procedure. In general, decidability depends on the restriction criteria.

\begin{table}[!t]
\centering
\resizebox{\textwidth}{!}{
\begin{tabular}{cccccc}
\begin{tabular}{c|c}
$\bot^{\mSfour}$ & $\top^{\mSfour}$\\ \hline
$0$ & $2$ \\ 
\end{tabular}
&
\begin{tabular}{c|ccc}
$x \to^{S4} y$ & 2 & 1 & 0 \\\hline
2 & $2$ & $1$ & $0$ \\
1 & $2$ & $1,2$ & $0$ \\
0 & $2$ & $1,2$ & $1,2$
\end{tabular}
&
\begin{tabular}{c|ccc}
$x \land^{S4} y$ & 2 & 1 & 0 \\\hline
2 & $2$ & $1$ & $0$ \\
1 & $1$ & $1$ & $0$ \\
0 & $0$ & $0$ & $0$
\end{tabular}
&
\begin{tabular}{c|ccc}
$x \lor^{S4} y $ & 2 & 1 & 0 \\\hline
2 & $2$ & $2$ & $2$ \\
1 & $2$ & $1,2$ & $1,2$ \\
0 & $2$ & $1,2$ & $0$
\end{tabular}
&
\begin{tabular}{c|c}
$x$ & $\neg^{S4} x$ \\\hline
2 & $0$ \\
1 & $0$ \\
0 & $1,2$
\end{tabular}
&
\begin{tabular}{c|c}
$x$ & $\Box^{S4} x$ \\\hline
2 & $2$ \\
1 & $0$ \\
0 & $0$
\end{tabular}

\end{tabular}
}
\caption{Non-deterministic truth-functions for $\mathbb{M}_{\mSfour}$.}\label{table:modalNmatrices}
\end{table}

The first decision procedure based on level valuations for $\mSfour$ was proposed by Gr\"atz in~\cite{DBLP:journals/logcom/Gratz22}. In this work, he defines a decidable notion of partial level valuation, which is applied to a three-valued Nmatrix, where the primitive truth-values are necessarily true ($2$), contingently true ($1$), and false ($0$). \Cref{table:allNmatrices} defines the language of the system, as well as the truth-values, while \Cref{table:modalNmatrices} defines the corresponding non-deterministic truth-functions\footnote{The truth-functions for conjunction and disjunction are due to~\cite{leme2025newIPL}.}. The restricted set of valuations can be defined as follows.

\begin{definition}[$\mathcal{L}^{\mSfour}$]\label{def:rsfour}
    Let $\mathcal{L}^{\mSfour}_0 = Val(\mathbb{M}_{\mSfour})$. $\mathcal{L}^{\mSfour}_{k+1}$ is the set of every $v \in \mathcal{L}^{\mSfour}_{k}$ such that, if $v(\alpha)=1$, then there is $w \in \mathcal{L}^{\mSfour}_{k}$ such that $w(\alpha) = 0$ and, for every $\beta$, $v(\beta) = 2$ implies $w(\beta) = 2$. Then,
    \[
    \mathcal{L}^{\mSfour} = \bigcap^{\infty}_{i=0} \mathcal{L}^{\mSfour}_i
    \]
\end{definition}

As for $\ipl$, the existence criteria can be decided locally and then generalized to all levels. In the case of $\mathbb{M}_{\mSfour}$, the dependency is created by the value $1$, which is designated but contingent. In the intended model, $v(\alpha) = 1$ means that $\alpha$ is true but possibly false. Now, if there is no valuation $w$ such that $w(\alpha) = 0$, then $\alpha$ is not contingently true but rather necessarily true. On the other hand, if there is such a $w$, then $v$ and $w$ must be consistent with both the rule of necessitation and axiom $\mfour$. In other words, if $v(\alpha) = 2$ (meaning that $\alpha$ is necessarily true), then $w(\alpha)$ must be designated and necessary. Hence, $v(\alpha) = 2$ implies $w(\alpha) = 2$.

\begin{definition}[RNmatrix for $\mSfour$]
    The RNmatrix for $\mSfour$ is given by $\mathcal{R}(\mathbb{M}_{\mSfour}) = \langle \mathcal{L}^{\mSfour}, \mathbb{M}_{\mSfour} \rangle$.
\end{definition}

\section{Tautologies in  RNmatrices as SMT Problems}\label{sec:encoding}
%!TEX root = main.tex
% !TEX spellcheck = en-US

This section presents an encoding of the validity problem for formulas in the
non-classical logics introduced in the previous section (intuitionistic, modal,
and paraconsistent) as an SMT problem. This encoding yields a new suite of
automated theorem provers for these logics that we benchmark in the next
section. The translation is performed automatically by our tool \tool, a
matrix-based theorem prover. In practice, \tool~is not a theorem prover itself,
but rather a translator: it takes as input a formula expressed in the given
logic using the \href{https://www.tptp.org/}{TPTP} format,  and it produces a
 file in the \href{https://smt-lib.org/}{SMT-LIB} format, which
can then be fed to any SMT solver.

In the rest of this section, we fix an arbitrary RNmatrix $\rnmatrix = \langle \cF, \ML{\Lan}\rangle$, assumed to be sound and complete (\ie, a formula is valid in $\Lan$ if and only if it is a tautology with respect to $\rnmatrix$).
%We also assume $\alpha$ to be an arbitrary formula in $\Lan$. 

The encoding of the validity problem for the formula \(\alpha\), using 
$\rnmatrix$, starts 
with the definition of a suitable multi-sorted signature. The distinct truth
values of the RNmatrix give rise to a sort \(\sortval\), populated with
pairwise distinct constants, one for each truth value. We also assume a sort
\(\sortrow\) for the rows and a constant $r_0$ of this sort, representing an arbitrary row in the matrix. 
Moreover, we consider the sort \(\sortcol\) for 
columns, populated with constants corresponding to each
subformula of the formula $\alpha$. 
Finally, we assume an uninterpreted function $mat$ of sort
$\sortrow\times \sortcol \to \sortval$ representing the 
matrix.

\begin{definition}[Signature]\label{def:sig}
Let $\alpha$ be a formula in $\Lan$ and $\rnmatrix$ an RNmatrix for $\Lan$. 
The associated first-order signature $\Sigma^{\alpha}_{\rnmatrix}$ includes: 
    Three sorts, namely $\sortval$, $\sortrow$, and $\sortcol$; 
    a constant $r_0$ of sort $\sortrow$;
    the constants $\{ c_{tv} \mid tv \in \V \}$ of sort $\sortval$; 
    the constants $\{c_{\beta} \mid \beta \in sub(\alpha)\}$ of sort $\sortcol$; 
    the (uninterpreted) function $mat$ of sort $\sortrow\times \sortcol \to \sortval$; and 
    the usual symbol $=$ for equality.
 \end{definition}

\paragraph{Goal and Constraints.} The idea is to check whether the (arbitrary)
row $r_0$  falsifies the formula $\alpha$ being tested. That is, if $mat(r_0,c_{\alpha})=
c_{tv}$ for a non-designated truth value $tv \in \V\setminus \D$. For that, we
have to impose constraints to: 1) guarantee that the truth values are all
pairwise distinct; 2) the value of $mat(r,c_\beta)$, for any row $r$,  
is allowed by the matrix according to the truth values associated to its subformulas; 
 and 3) the rows of the matrix are ``valid'',  \ie,
they are in the set $\cF \subseteq Val(\ML{\Lan})$.

\begin{definition}[Constraints]\label{def:cons}
Let $\alpha$ be a formula in $\Lan$, $\rnmatrix$ an RNmatrix for $\Lan$, 
and 
$\Sigma^{\alpha}_{\rnmatrix}$ as in~\Cref{def:sig}. 
We define  $\Phi^{\alpha}_{\rnmatrix} \defsym \Phi_{r_0} \wedge \Phi_{tv} \wedge \Phi_{C} \wedge \Phi_{M} \wedge \Phi_{F}$ where:
\begin{itemize}
    \item $\Phi_{r_0}\defsym \bigvee\limits_{tv\in \V\setminus \D }mat(r_0, c_{\alpha}) = c_{tv}$;
    \item $\Phi_{tv} \defsym  \bigwedge\{ \neg (c_{tv} = c_{tv'})\mid tv, tv' \in \V \mbox{ and } tv \neq tv'\}$;
    \item $\Phi_{C} \defsym  \bigwedge\{ \neg (c_{\beta} = c_{\beta'})\mid \beta, \beta' \in sub(\alpha) \mbox{ and } \beta \neq \beta'\}$;
    \item $\Phi_{M} \defsym  \forall r:\sortrow.\bigwedge\limits_{\beta\in sub(\alpha)}(cons_\alpha(\beta,r))$ where $cons_\alpha(\beta,r)$ is defined as:
\[
    \begin{array}{lll}
        cons_\alpha(\beta,r) & \defsym &  \left\{
        \begin{array}{l}
            \bigvee\limits_{tv\in \V} mat(r,c_\beta)= c_{tv}  \mbox{ (which is equivalent to $true$) if  $\beta$ is atomic  }\\
            f_\star(mat(r,c_\gamma), mat(r,c_\beta)) \mbox{ if $\beta = \star \gamma$ for an unary connective $\star$}\\
            f_\star(mat(r,c_\gamma), mat(r,c_\delta), mat(r,c_\beta)) \mbox{ if $\beta = \gamma \star \delta$ for a binary connective $\star$}
        \end{array}
        \right. \\
            f_\star(x,y)   &  \defsym  &\bigwedge\limits_{tx\in \V}( x = c_{tx} \Rightarrow \bigvee\limits_{ty\in \tilde\star(tv)} y = c_{ty}) \\
            f_\star(x,y,z) &  \defsym  &\bigwedge\limits_{tx,ty \in \V} ((x = c_{tx} \wedge y = c_{ty}) \Rightarrow \bigvee\limits_{tz\in \tilde\star(tx,ty)} z = c_{tz} )
    \end{array}
\]
\item $\Phi_F$ is a formula of the form $\forall r:\sortrow.\Phi'$ where $\Phi'$ is the first order formula
    characterizing the criterion $\cF$. 

\end{itemize}
\end{definition}

In the case of atomic propositions, the corresponding columns 
may take any truth value in $\V$.
For unary conectives, $f_\star(x,y)$ imposes the constraint
asserting that the possible values of $y$ are determined
by those of $x$ according to the  RNmatrix. 
For instance, for the paraconsistent logic $C_1$, 
 $f_{\neg}$ is defined as:
\[
    f_{\neg}(x,y) \defsym (x=c_F \Rightarrow y=c_T) \wedge (x=c_t \Rightarrow (y=c_t \vee y=c_T)) \wedge (x=c_T \Rightarrow y=c_F)
\]
Similarly, for the binary connectives, we define suitable 
functions enforcing the corresponding values for the RNmatrix. For instance, 
after some simplifications, we  define the restriction for disjunction as:
\[
    f_{\vee}(x,y,z) \defsym ( (x=c_t \vee y=c_t)\Rightarrow z=c_t) \wedge ((x=c_T \vee y=c_T)\Rightarrow z=c_T) \wedge
    ((x=c_F \wedge y=c_F)\Rightarrow z=c_F)
\]

The formula $\Phi_F$ for $C_1$  asserts that, for each row, 
a $t$ value in a column $\beta$ implies that the column $\beta\wedge \neg \beta$ necessarily takes the value $T$: 
\begin{equation}\label{eq:c1-correct}
    \forall r:\sortrow.\bigwedge\limits_{(\beta\wedge\neg \beta) \in subs(\alpha)} (mat(r,c_\beta) =c_t) \Rightarrow mat(r,c_{(\beta\wedge\neg \beta)}) = c_T
\end{equation}

\begin{theorem}[Correctness]
Let $\alpha$ be a formula in $\Lan$, $\rnmatrix$ an RNmatrix for $\Lan$
and 
$\Phi^{\alpha}_{\rnmatrix}$ be as in~\Cref{def:cons}. 
Then, $\alpha$ is valid in $\Lan$ iff $\Phi^{\alpha}_{\rnmatrix}$ is unsatisfiable in FOL. 
\end{theorem}
\begin{proof}
    By soundness and completeness of $\rnmatrix$, 
    $\alpha$ is valid iff 
    there is no row that assigns a non-designated value to the column corresponding to $\alpha$ and, 
    by construction, this is the case iff $\Phi^{\alpha}_{\rnmatrix}$ is unsatisfiable. 
\end{proof}

\paragraph{Using SMT solvers.} The above encoding allows us to use any
first-order theorem prover to decide the validity  of a formula $\alpha$ in a given
logic $\Lan$ via the corresponding RNmatrix. However, for better efficiency,
one can exploit specific theories supported by SMT solvers. In particular, the
SMT-LIB file produced by \tool{} makes use of the theories of arrays and
datatypes. The theory of arrays allows us to represent the matrix with
an arbitrary number of rows in a natural way. Although uninterpreted functions could also be
used for this purpose, our experiments
indicate that arrays lead to better performance. Moreover, the theory of
datatypes is used to define the pairwise distinct constants representing truth
values and columns, and formulas $\Phi_{t_v}$ and $\Phi_{C}$  in~\Cref{def:cons} are not needed
(but internalized in the datatype theory). 
Alternative encodings, for instance using  bit-vectors, are also possible
and are currently under evaluation.

Considering the logic $C_1$, let us show step by step the SMT file produced by \tool, that follows
the formula in~\Cref{def:cons}.  
In this logic, we have 3 different values, leading to the following declaration
 (the ``0'' in ``\texttt{Val 0}'' specifies that \texttt{Val} is a scalar type, and comments in SMT-LIB
 start with ``\texttt{;}''): 

\begin{lstlisting}
; Truth values for C1
(declare-datatypes ((Val 0)) (((TT) (FF) (tt))))
\end{lstlisting}

Next we define the sort for the columns of the matrix which is populated
with (different) constants, one for each subformula of the formula being tested. 
For instance, for the formula $p \vee \neg p$, we obtain:

\begin{lstlisting}
(declare-datatypes ((Col 0)) (( (v2)      ;  (p | ~ p)
                                (v1)      ;  ~ p
                                (v0))))   ;  p
\end{lstlisting}

\noindent Note that the formula being tested for validity is represented by the 
expression \texttt{v2}. 

Rows in the matrix are represented as elements of a new sort \texttt{Row}. 
\begin{lstlisting}
(declare-sort Row 0)
\end{lstlisting}

A valuation  is a mapping from elements of sort \texttt{Col} to elements
of sort \texttt{Val}, and a matrix 
is represented as an array (with no restriction on the size of the structure):  

\begin{lstlisting}
; A matrix mapping a column of a row into a truth value
(declare-fun mat () (Array Row (Array Col Val)))
\end{lstlisting}

Next the row $r_0$ is defined together with the  goal, 
stating that $mat[r_0][\alpha]$ takes a non-designated value:

\begin{lstlisting}
(declare-fun r_0 () Row)
(assert (= (select (select mat r_0) v2) FF))   ; Falsifying v2 (p | ~p)
\end{lstlisting}

Now we add constraints enforcing that the values allowed 
for a given column representing a compound formula 
are only those allowed by the multifunction associated to 
its main connective. 
This leads to two constraints: one for $\neg p$
and one for $p \vee \neg p$. Below the case
of the negation, where we use 
the if-then-else operator (\texttt{ite}) in SMT-LIB
to define $f_{\neg}$: 

\begin{lstlisting}
(assert (forall ((r_x Row))
    (ite (= (select (select mat r_x) v0) FF)              ; IF      v0 =  F
         (= (select (select mat r_x) v1) TT)              ; THEN    v1 =  T
         (ite (= (select (select mat r_x) v_0) tt)        ; ELSE-IF v0 =  t
              (distinct (select (select mat r_x) v1) FF)  ; THEN    v1 != F 
              (= (select (select mat r_x) v_1) FF)))))    ; ELSE    v1 =  F 
\end{lstlisting}

For this particular formula, there are no occurrences of subformulas
of the form $\alpha \wedge \neg \alpha$ and then,
no further constraints are added for the correctness criteria.

On this input, the SMT solver returns \texttt{unsat}, and we conclude that the formula is valid:
it is not possible to find a ``valid'' row (according to the elimination criteria of the RNmatrix) assigning a non-designated value to the column 
of the formula being tested.

\section{Implementation and Benchmarks}\label{sec:bench}
In this section we briefly describe the architecture of \tool, some
``optimizations'' in the produced SMT-LIB files,
and report the results on benchmarks. 
\tool~has been implemented in
C++ and it is freely available at~\cite{tool}. It uses
\href{https://www.antlr.org/}{ANTLR} for parsing the
\href{https://www.tptp.org/}{TPTP} files and
\href{https://www.boost.org/}{Boost} for executing processes in parallel (see below).
\footnote{A new implementation of \tool~in Rust is currently under development.
This will simplify the compilation and distribution of the tool through Cargo,
Rust’s package manager.} All the experiments reported here
 can be reproduced with
the scripts available in the tool. The plots shown here are based on executions on a
HP Dragonfly laptop 
with a processor Intel Core i7-1365U
10-core @ 5.2 GHz, 32 GiB memory, running Fedora Linux 43 and 
 Z3 version 4.15.8.

As mentioned before, \tool{} takes as input a TPTP
file (containing the conjecture under consideration) and produces an SMT-LIB file
encoding the problem as described in the previous section. In principle, any
SMT solver can be used to process this output. However, all the experiments
reported here were conducted using Z3, as it consistently gave us 
better results in the considered benchmarks.

\tool{} defines an abstract class \texttt{RNMatrix}  whose methods
need to be override by the logic at hand. 
Such methods correspond to the different formulas in 
~\Cref{def:cons}. The implementation
of the subclasses of \texttt{RNMatrix} 
are free to choose 
 any representation for the sorts (using datatypes, for instance), 
the matrix (uninterpreted functions or arrays) and 
the truth values (datatypes or integers). 
Below an example of one of these methods for the logic $C_1$, declaring
the truth values of the logic:

\begin{lstlisting}
std::string LogicC1::addSorts(){
    return  "(declare-datatypes ((VAL 0)) (((TT) (FF) (tt))))\n";
}
\end{lstlisting}

%!TEX root = main.tex
% !TEX spellcheck = en-US

\subsection{Benchmarking paraconsistent logics}

Using the RNmatrix in~\Cref{mat:cn}, \tool~becomes a prover for the entire
hierarchy $C_n$ of paraconsistent logics. In these logics, the criterion for
determining the validity of a row depends only on the values taken by certain
columns within \emph{the same row} (see~\Cref{eq:c1-correct} in the previous
section). This significantly simplifies the scenario: if row $r_0$ is
falsifiable, there is no need to show the existence of other rows to check
that $r_0\in \cF$, i.e., that $r_0$ is  a ``valid'' row in the matrix. 
Hence, there is no need to define a matrix with multiple rows (since checking a
single row suffices), nor to introduce a sort for columns.
 Instead, we declare,
for each subformula $\beta$ of $\alpha$, a constant $c_\beta$ of sort $\sortval$. Compare
the actual output of \tool~below with the one shown in the previous section for
the formula $p \vee \neg p$, where neither arrays nor quantification over rows
are required:

\begin{lstlisting}
(declare-datatypes ((VAL 0)) (((TT) (FF) (tt))))
(declare-fun v2 () VAL) ; (p | ~ p)
(declare-fun v1 () VAL) ; ~ p
(declare-fun v0 () VAL) ; p
(assert (= v2 FF))      ;  Main formula falsified
;; Rules of the matrix
(assert (ite (and (= v0 FF) (= v1 FF)) (= v2 FF)... ; (v2 = v0 | v1) IF V0=V1 = F THEN V2=F ...
(assert (ite (= v0 FF) (= v1 TT) ...                ; (v1 = not v0)  IF v0=F THEN v1=T ...
(check-sat)
\end{lstlisting}

\begin{figure}
\begin{subfigure}[c]{0.25\textwidth}
    \centering
\resizebox{.5\textwidth}{!}{
\begin{tabular}{ll}
\toprule
\rowcolor{gray!20}
$C_n$ & time \\
\midrule
$c_{20}$ & 17s \\
$c_{40}$ & 31s \\
$c_{100}$ & 3m16s \\
\bottomrule
\end{tabular}
}
\caption{\label{fig:cn:1}}
\end{subfigure}
\begin{subfigure}[c]{0.25\textwidth}
    \centering
\resizebox{.95\textwidth}{!}{
\begin{tabular}{lr}
\toprule
\rowcolor{gray!20}
group & subform (n=500) \\
\midrule
fifth & 5508 \\
first& 4403 \\
firstWrong& 4402 \\
 fourth& 5502 \\
 fourthWrong& 5501 \\
 second & 66503 \\
 sixth& 10004 \\
\bottomrule
\end{tabular}
}
\caption{\label{fig:cn:2}}
\end{subfigure}
~
\begin{subfigure}[c]{0.5\textwidth}
\centering
    \includegraphics[scale=0.35]{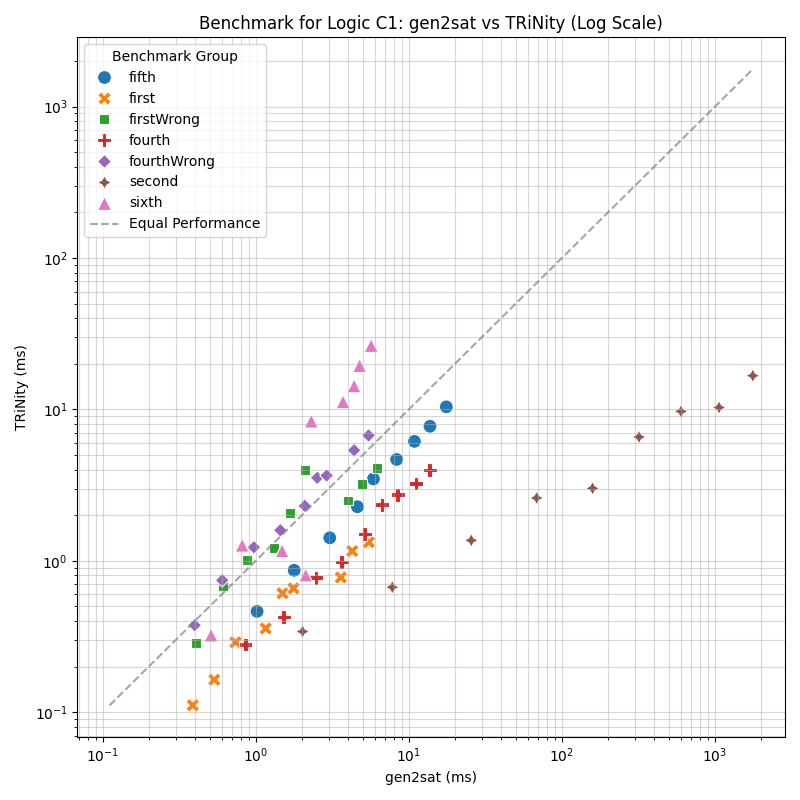}
    \caption{\label{fig:cn:3}}
\end{subfigure}
    \caption{
        \Cref{fig:cn:1} shows the 
        time taken by \tool~for completing the 45 entries of  \texttt{Cn\_propag}
        for different logics $C_n$;  \Cref{fig:cn:2} shows the number of subformulas 
for the instance $i=500$ for the different families in the benchmarks for $C_1$; 
\Cref{fig:cn:3} compares 
\tool~vs~gen2sat for the benchmarks on $C_1$.
Items below the diagonal are instances that \tool\ solved faster. 
\label{fig:cn}}
%\captionof{figure}{Computation times}
\end{figure}

We consider the six benchmark families described in~\cite{neto2006effective} and~\cite{neto2009towards} for logic $C_1$. 
Furthermore, we introduce a new family of benchmarks for any logic $C_n$, called
\texttt{Cn\_propag}, which is based on the propagation of consistency over
conjunction, disjunction, and implication
(see~\cite{CarnielliConiglio2016mbC}).

Let us begin with \texttt{Cn\_propag},
that contains three families $\mathit{conj}_i$, $\mathit{disj}_i$ and $\mathit{imp}_i$ where
$1\leq i\leq 15$ for a total of 45 formulas.
Formulas in this family have the form
$\alpha^n \land \beta^n \to (\alpha \# \beta)^n$, where $\# \in \{ \land, \lor,
\to \}$. Recall from~\Cref{mat:cn} that $\alpha^n$ encodes the consistency
depth of $\alpha$. Hence, \texttt{Cn_propag} states that the consistency of
smaller formulas implies the consistency of larger formulas. However, each
$\cn$ is not only capable of blocking explosion up to a certain level, but also
of blocking the propagation of consistency up to that level. For example,
$\alpha^1$ and $\beta^1$ imply that $(\alpha \to \beta)^1$ is true in $C_1$,
but not in any of the logics $C_i$ with $i > 1$.

 The formula $\alpha_{15}$, for each family $\alpha$, 
contain  182 subformulas.  
We have used this benchmark as a ``sanity check'' to ensure that the prover 
determines as valid only the formulas $\alpha_i$ in the logic $C_j$ whenever $j \leq i$.  
Checking that the logic $C_{10}$ validates the corresponding $3 \times 10$ 
formulas, while rejecting the remaining $3 \times 5$ formulas, takes approximately 
13 seconds. \Cref{fig:cn:1} presents additional instances for $n$ in  $C_n$  when checking all the 
45 formulas in this benchmark. As expected, performance degrades as $n$ increases 
in $C_n$, since more truth values are introduced.

For logic $C_1$, we benchmark \tool~against \texttt{gen2sat} 
 with a selected subset of the families of benchmark proposed in~\cite{neto2006effective} and~\cite{neto2009towards}. The results are shown
in~\Cref{fig:cn:3}, where each family was tested with instances ranging from 100
to 500 in increments of 50, yielding a total of 9 instances per family. The
number of subformulas in each case, for the instance $n = 500$, is shown in
\Cref{fig:cn:2}. Note that 
\texttt{gen2sat} performs better on the \texttt{sixth} family (5.6 vs.\ 26
seconds for $n = 500$), while \tool~significantly outperforms \texttt{gen2sat}
on the \texttt{second} family (16 vs.\ 1766 seconds for $n = 500$). For the
remaining families, the differences are mild; for instance, 5.4 vs.\ 1.3
seconds for the \texttt{first} family when $n = 500$.

%!TEX root = main.tex
% !TEX spellcheck = en-US

\subsection{Benchmarking $\mSfour$}

 For the modal logic $\mSfour$, we evaluate our tool using a selected subset of
 benchmarks from the Logic Workbench (LWB). We compare the performance of
 $\mettel$, $\ksp$, and \tool. $\mettel$~\cite{tishkovsky2012tableau} is a framework for generating tableau-based theorem provers from user-friendly specifications of proof systems. $\ksp$~\cite{ijcai2017p694} is an automated theorem prover implementing a resolution-based calculus for modal logic $\mK$ and its extensions, supporting additional configuration options (e.g., strategy and ordering settings). For the benchmarks with $\ksp$, we used the configuration file \texttt{ordered.conf}, which provided the best overall performance in our tests, together with the $\mSfour$ configurations described in~\cite{nalon2022local}.

 We use the 3-valued RNmatrix in~\Cref{sec:modal}. We have several choices for
 implementing the  class \texttt{LogicS4} (that extends \texttt{RNMatrix}). For instance, 
 to represent the truth values, we can use datatypes or just integers (0,1 and 2). 
 This latter representation gave us better results. 
 More importantly, 
 it  is well known~\cite{DBLP:conf/ecai/ArecesGHR00} that a modal formula $\alpha$ is satisfiable iff 
 it is satisfiable in a model 
 bounded 
 by the modal-depth of $\alpha$ ($\delta_\alpha$ below). %\co{Do we have a reference for that?} \rl{Maybe \url{https://tinyurl.com/227y3xn9} ?}
 Hence, our translation introduce the 
following function and constraints: 
\begin{lstlisting}
(declare-fun depth (Row) Int)
(assert (= (depth r_0) 0))                          ; depth of r0 is 0
(assert (forall ((r_x Row)) (<= (depth r_x)  $\delta_A$ ))) ; the depth of a row is less than $\delta_A$
\end{lstlisting}
Moreover, in the definition of the correctness criteria, each time the existence of another row $r_y$
is needed to support the ``validity'' a row $r_x$ ($r_x \in \cF$), the depth of $r_y$ is assumed to be 
the depth of $r_x$ plus one: 
\begin{lstlisting}
(assert (forall ((r_x Row) (c_x Col))                
   (=> (= (select (select mat r_x) c_x) 1 )        ; IF mat[rx][cx] = 1
       (exists ((r_y Row))                         ; THEN there exists ry
        (and (= (depth r_y) (+ 1 (depth r_x))) ... ;      where depth(ry) = depth(rx)+1 and ...
\end{lstlisting}

This allows us to control that only models of a given depth are considered. 

We also implemented a second instance of the class \texttt{RNMatrix} that we
call \texttt{LogicS4Bound}, where the sort \texttt{Row} is declared as a
datatype with $n$ constants. In this case, only 
models of up to $n$ rows/worlds are considered.
If the SMT
solver returns \texttt{sat} for one of these \emph{bounded} models, we have a
countermodel of a fixed size $n$. However, if the SMT solver returns
\texttt{unsat}, we cannot conclude that the formula is valid. 
For non-valid formulas, this can be useful to quickly find coutermodels. 

\begin{figure}
\begin{minipage}[c]{0.4\textwidth}

    \resizebox{.7\textwidth}{!}{
    \begin{tabular}{lccc}
\toprule
Prover & $\ksp$ & $\mettel$ & \tool \\
Family &  &  &  \\
\midrule
s4-45-n & \textbf{8} & 0 & 3 \\
s4-45-p & \textbf{29} & 0 & 13 \\
s4-branch-n & \textbf{4} & 2 & 1 \\
s4-branch-p & \textbf{7} & 2 & 6 \\
s4-grz-n & \textbf{7} & 1 & \textbf{7} \\
s4-grz-p & \textbf{46} & 0 & 20 \\
s4-ipc-n & 38 & 3 & \textbf{100} \\
s4-ipc-p & \textbf{6} & 1 & \textbf{6} \\
s4-md-n & 8 & 12 & \textbf{40} \\
s4-md-p & \textbf{4} & 2 & \textbf{4} \\
s4-path-n & 12 & 2 & \textbf{19} \\
s4-path-p & \textbf{13} & 1 & 3 \\
s4-ph-n & 5 & 3 & \textbf{10} \\
s4-ph-p & \textbf{5} & 3 & \textbf{5} \\
s4-s5-n & \textbf{6} & 2 & 2 \\
s4-s5-p & \textbf{10} & 1 & 1 \\
s4-t4p-n & 4 & 0 & \textbf{7} \\
s4-t4p-p & 7 & 0 & \textbf{8} \\
\bottomrule
\end{tabular}
}
\end{minipage}
\hfill
\begin{minipage}[c]{0.6\textwidth}
    \includegraphics[scale=0.5]{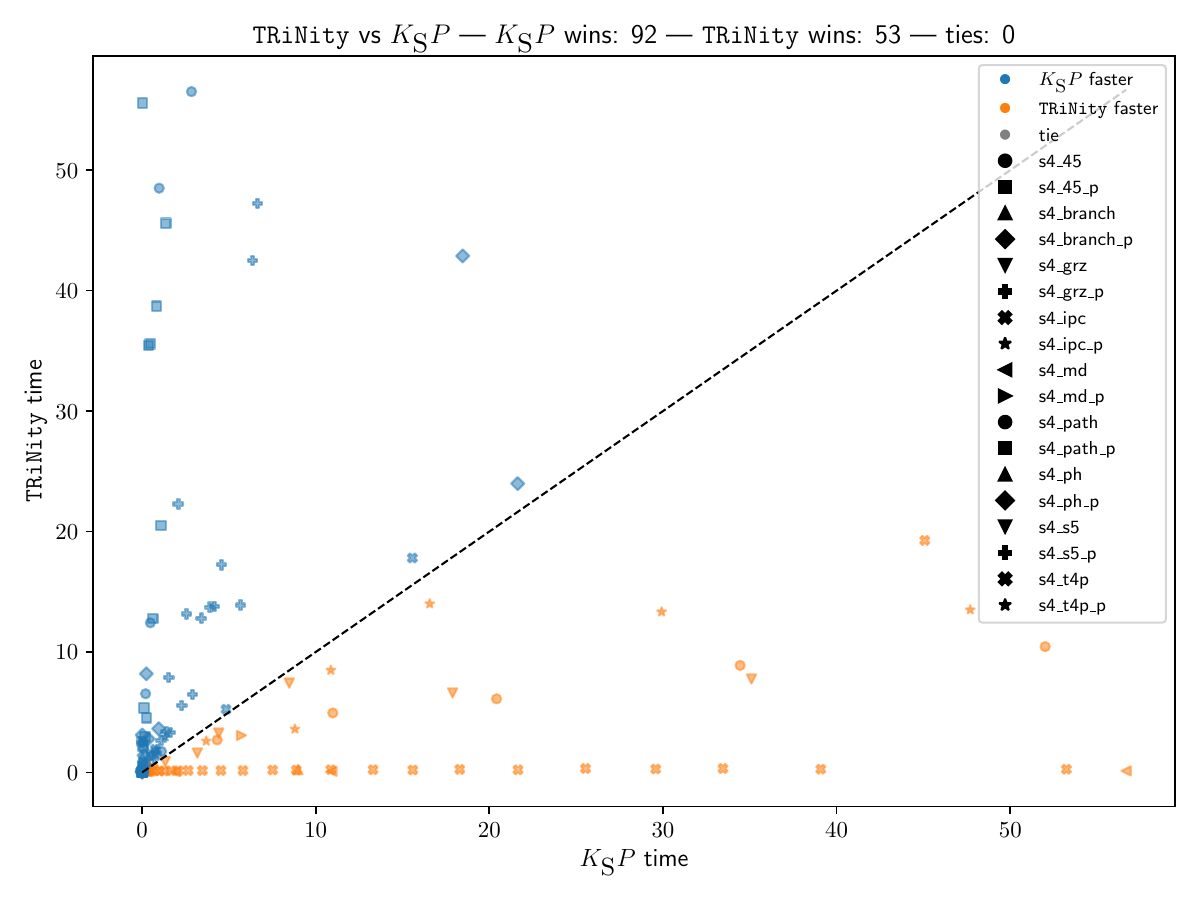}
\end{minipage}

    \caption{Benchmarks for logic S4. On the left, the maximum number of instances solved in 1 minute for each solver. 
    On the right, we compare the time taken by \tool~and $\ksp$ to solve all the instances.\label{fig:s4}}
\end{figure}

In~\Cref{fig:s4} we compare $\ksp$, $\mettel$ and \tool~running in parallel instances
of Z3 with the SMT-LIB file generated by \texttt{LogicS4} and
at most 4 instances (with different bounds) of \texttt{LogicS4Bound}. On the left of the figure, we
show the maximum number of instances solved in 1 minute for each solver and
family. In bold face, the winner. The results show that out tool, for some of
the families, is able to solve more instances in the given time window. We note
that for the positive instances (suffix ``p'' in the family), the only instance of Z3
that can decide the formula  is the one generated by \texttt{LogicS4} (i.e., the answer
of the ``bounded'' provers cannot be considered). In the
negative instances, in some cases, it is \texttt{LogicS4Bound} that disproves
the formula. On the right of the figure, we compare the time taken by
our tool and $\ksp$ to solve all the instances. In 92 instances, $\ksp$
outperformed \tool, and in 53 instances \tool~outperformed $\ksp$. 

%!TEX root = main.tex
% !TEX spellcheck = en-US

\subsection{Benchmarking Intuitionistic Propositional Logic}\label{sec:benc-ipl}

For $\ILp$, we consider the propositional fragment of the ILTP library
(version~1.1.2). In addition, we include two extra benchmark families, namely
\texttt{KLE} and \texttt{EC}. \texttt{KLE} is a set of theorems (and two non-theorems) extracted from Kleene's book ``Introduction to Metamathematics'' by Valeria de Paiva and Giselle Reis. \texttt{EC} is one of the families of valid formulas in $\ipl$ used for benchmark in~\cite{claessen2015sat} and~\cite{fiorentini2021efficient}.

The class \texttt{LogicIPL} implementes the RNmatrix for $\ILp$ with two values.
In this case, instead of defining those values as another datatype (or as an
integer), we use the Boolean type of the SMT. Similar to the case of S4, we
also implemented instances of the class \texttt{RNmatrix} that control the depth of the
formula (in this case nesting of implications and negations) and bounded
versions of it. 
The results are in~\Cref{tab:ipl}. We
compare $\intuitR$, $\ksp$ (using the G\"odel translation of IPL into S4), $\mettel$
and two provers based on our tool: \tool(ipl) implementing the matrix for IPL, and \tool(ipl-all)
that runs 
 Z3 with the output 
of \texttt{LogicIPL} in parallel with 
the output of at  most 4 instances (with different bounds) of \texttt{LogicIPLBound}.
$\intuitR$ definitly outperforms all the others. 

We note that $\intuitR$ (see~\cite{fiorentini2021efficient}) is a solver for $\ipl$ that relies on a combination of SAT solvers and Kripke semantics. It applies the strategy introduced in~\cite{claessen2015sat} for \texttt{intuit}, which consists of separating the validity problem into two parts, one of which can be solved using classical logic. To this end, they introduce a clausification procedure. The idea is that the hardest part of the problem is handled by the SAT solver, which benefits the efficiency of the intuitionistic solver. 
We also note that $\ksp$ did not perform very well in the family \texttt{EC}. This may be related to the growth of the modal depth in the formulas produced by the translation. $\tool$, on the other hand, performs particularly well on this family of formulas. In \texttt{EC}, the number of atoms, conjunctions, and disjunctions increases. However, the formulas in this family always contain only two implications. For this reason, the SMT problem produced by $\tool$ is almost a classical SAT problem in the semantics, which can then be solved efficiently by Z3. 
The results for the 
harder family SYJ 
are very similar for $\ksp$ and \tool(ipl),
and \tool(ipl-all) was able to solve more (negative) instances
in the given timeout in this case.

\begin{table}
    \centering
    \begin{tabular}{lcccccc}
\toprule
Family & $\intuitR$ & $\ksp$ & $\mettel$ & \tool(ipl) & \tool(ipl-all) \\
\midrule
EC(100) & 100 & 44 & 100 & 100 & 100 \\
LCL(2) & 2 & 2 & 2 & 2 & 2 \\
KLE(88) & 88 & 88 & 79 & 88 & 88 \\
SYN(20) & 20 & 20 & 19 & 19 & 19\\
SYJ(252) & 242 & 116 & 68 & 113 & 134 \\
\bottomrule
\end{tabular}
    
    \caption{Benchmarks for $\ILp$. Instances solved in up to 2 minutes for each prover.\label{tab:ipl}}
\end{table}

\section{Concluding Remarks}\label{sec:conclusion}
%!TEX root = main.tex
% !TEX spellcheck = en-US

We have introduced a new, efficient approach to implementing decision procedures based on RNmatrices. We find that $\tool$ achieves competitive performance across all the implemented logics, especially for the paraconsistent logics $\cn$. This approach also has the advantage of relying on SMT solvers without being tied to any particular implementation. Moreover, we believe that this contribution provides a general framework that can be extended and adapted both for implementing known RNmatrices and for discovering new ones. In particular, we plan to investigate RNmatrices for modal intuitionistic logic, for which $\tool$ can be employed from the outset.

Besides developing new RNmatrices, it would be interesting to find reductions of known RNmatrices. For instance, the RNmatrix for $\ipl$ that we have presented is a $2$-valued reduction of an equivalent $3$-valued RNmatrix for the same logic. One may ask whether there exists an elegant reduction of the $3$-valued decision procedure for $\mSfour$ to a $2$-valued version. Similar questions can be raised for other known RNmatrices, such as those we have proposed in~\cite{DBLP:conf/tableaux/LemeOPC25}.

Finally, exploring alternative encodings (integers, datatypes, bit-vectors) and systematically comparing SMT solvers constitutes a natural next step for further improving performance and scalability.

\paragraph{Acknowledgments} Leme thanks the support of the
São Paulo Research Foundation (FAPESP, Brazil) through the PhD scholarship
grant \#21/01025-3 and the Thematic Project RatioLog \#20/16353-3. This work was funded by French-Brazilian CAPES-COFECUB Research
Programme ``Logic and Intelligibility of Computational Processes'' (2024-2027). Olarte gratefully acknowledges also the
support from the NATO Science for Peace and Security
Programme through grant number G6133 (project SymSafe) and the SGR
project PROMUEVA (BPIN 2021000100160) under
the supervision of Minciencias Colombia. Pimentel is supported by the Leverhulme Trust grant RPG-2024-196 and
has received funding from the European Union's Horizon 2020 research and
innovation programme under the Marie Sk\l odowska-Curie grant agreement Number
101299559.
The authors are grateful for the valuable discussions with Marcelo E. Coniglio and Yoni Zohar, and for the insightful comments of the anonymous reviewers.

%\begin{itemize}
%    \item This framework can be used to ``validate'' new matrix based semantics...
%    \item This also provides new theorem provers for logics we are 
%        investigating (modal + intuitionistic)
%    \item Future: The quest of ``reducing'' the values in the matrix and ``simplifying''
%        the correctness check in the matrix.  {\color{blue} RL: fiz até aqui... }
%    \item Future: experimenting with other ways of encoding the problem (int vs datatypes, 
%        bit vectors, etc) and test other SMT solvers. 
%
%\end{itemize}

\bibliographystyle{eptcs}
\bibliography{biblio}
\end{document}